\documentclass[aps,pra,twocolumn,superscriptaddress]{revtex4-1}
\usepackage{amssymb,amsmath,amsthm}
\usepackage{graphicx}
\usepackage{bm,bbm}
\usepackage[bookmarks=false]{hyperref}
\hypersetup{colorlinks=true,citecolor=blue,linkcolor=red,%
urlcolor=blue,pdfstartview=FitH,bookmarksopen=true}
\usepackage{comment}

\usepackage{algorithm}
\usepackage{algcompatible}

\newcommand{\ket}[1]{\mbox{$ | #1 \rangle $}}
\newcommand{\bra}[1]{\mbox{$ \langle #1 | $}}

\newcommand{\tr}{\mathrm{Tr}}
\newcommand{\cP}{\mathcal{P}}
\newcommand{\cM}{\mathcal{M}}
\newcommand{\cC}{\mathcal{C}}

\newcommand{\cZ}{\mathcal{Z}}

\newcommand{\cB}{\mathcal{B}}
\newcommand{\cH}{\mathcal{H}}
\newcommand{\cE}{\mathcal{E}}
\newcommand{\cK}{\mathcal{K}}
\newcommand{\cI}{\mathcal{I}}
\newcommand{\cX}{\mathcal{X}}
\newcommand{\cS}{\mathcal{S}}

\newcommand{\I}{\mathrm{i}}

\newtheoremstyle{note}
  {\topsep/2}              	
  {\topsep/2}            	
  {}                        
  {\parindent}             	
  {\itshape}                
  {.---}                    
  {0pt}                     
  {\thmname{#1}\thmnumber{ \itshape#2}\thmnote{ (#3)}} 

\newtheorem{theorem}{Theorem}

\newtheorem{proposition}[theorem]{Proposition}

\theoremstyle{definition}

\theoremstyle{remark}

\begin{document}
\title{Efficient verification of quantum processes}

\author{Ye-Chao Liu}
\affiliation{Key Laboratory of Advanced Optoelectronic Quantum Architecture and Measurement of
Ministry of Education, School of Physics, Beijing Institute of Technology, Beijing 100081, China}

\author{Jiangwei Shang}
\email{jiangwei.shang@bit.edu.cn}
\affiliation{Key Laboratory of Advanced Optoelectronic Quantum Architecture and Measurement of
Ministry of Education, School of Physics, Beijing Institute of Technology, Beijing 100081, China}

\author{Xiao-Dong Yu}
\email{Xiao-Dong.Yu@uni-siegen.de}
\affiliation{Naturwissenschaftlich-Technische Fakult\"at, Universit\"at Siegen,
Walter-Flex-Str. 3, D-57068 Siegen, Germany}

\author{Xiangdong Zhang}
\email{zhangxd@bit.edu.cn}
\affiliation{Key Laboratory of Advanced Optoelectronic Quantum Architecture and Measurement of
Ministry of Education, School of Physics, Beijing Institute of Technology, Beijing 100081, China}

\date{\today}
%

\begin{abstract}
Quantum processes, such as quantum circuits, quantum memories, and quantum
channels, are essential ingredients in almost all quantum information
processing tasks. However, the characterization of these processes remains a daunting task due to the exponentially increasing amount of resources required
by traditional methods. Here, by first proposing the concept of
quantum process verification, we establish two efficient and practical protocols
for verifying quantum processes which can provide an exponential improvement
over the standard quantum process tomography and a quadratic improvement
over the method of direct fidelity estimation.
The efficacy of our protocols is illustrated with the verification of various
quantum gates as well as the processes of well-known quantum circuits.
Moreover, our protocols are readily applicable with current experimental techniques
since only local measurements are required.
In addition, we show that our protocols for verifying quantum
processes can be easily adapted to verify quantum measurements.
\end{abstract}

\maketitle
%

\section{Introduction}%
Quantum processes, such as quantum circuits, quantum memories, and quantum channels,
are a broad class of transformations that a quantum mechanical system can undergo.
Naturally, the characterization and identification of these quantum processes become
an indispensable task in many fields of quantum information processing 
\cite{Eisert.etal2019}, including quantum communication \cite{Wilde2013}, 
quantum computation \cite{QCQI2010}, quantum
metrology \cite{Giovannetti.etal2011,Toth.Appelaniz2014}, and more.
Take quantum gates, or more generally quantum circuits, as an example.
Efficient and reliable characterization of quantum circuits
plays a vital role to guarantee the correctness of the computation results of
quantum devices.
Recently, several commercial institutions have claimed great progress in the
development of quantum computers with tens of qubits being built in practice.
However, the standard quantum process tomography (QPT) method
\cite{Chuang.Nielsen1997,Poyatos.etal1997} which requires measurements of the
order of $O(4^n)$ for an $n$-qubit system has already become powerless for
these intermediate-size quantum devices.

Hence, a lot of effort has been devoted to searching for alternative nontomographic methods.
Along this research line there are, for instance, compressed sensing QPT
\cite{Gross.etal2010,Flammia.etal2012b,Kliesch.etal2019},
direct fidelity estimation (DFE) \cite{Flammia.Liu2011, daSilva.etal2011}, and
randomized benchmarking (RB)
\cite{Emerson.etal2005,Levi.etal2007,Knill.etal2008,%
Dankert.etal2009,Magesan.etal2011}, as well as accreditation protocols
\cite{Samuele.etal2019}.
However, all these approaches either still require a large amount of
resources, or are not generally applicable in many scenarios of practical interest.

Similar problems also exist for the characterization of quantum states,
where many efficient methods are proposed to overcome the resource-inefficiency
problem with tomography.
Among them, quantum state verification (QSV)~\cite{Pallister.etal2018} has drawn
much interest due to its various nice properties including a high efficiency and
the requirement for local measurements only.
In short, QSV is a procedure for gaining confidence that the output of a quantum
device is a particular state by employing local measurements.
Many kinds of bipartite or multipartite quantum states
\cite{Morimae.etal2017,Pallister.etal2018,Takeuchi.Morimae2018,%
Hayashi.etal2006,Zhu.Hayashi2019b,Dimic.Dakic2018,Yu.etal2019,%
Li.etal2019,Wang.Hayashi2019,Zhu.Hayashi2019a,Liu.etal2019b,Zhang.etal2019,%
Zhu.Hayashi2019c,Zhu.Hayashi2019d,Li.etal2019b,Saggio.etal2019}
can be verified efficiently or even optimally by QSV.
In general, a QSV protocol $\Omega$ for verifying the target state \ket{\phi}
has the form
\begin{equation}
  \Omega = \sum_ip_i\Omega_i\,,
\end{equation}
where $\{\Omega_i,\openone-\Omega_i\}$ are random ``pass-or-fail'' tests,
which are implementable with local measurements and satisfy that $\tr(\Omega_i
\ket{\phi}\bra{\phi})=1$ for all $i$.
In the case where all $N$ states pass the test, we achieve the
confidence level $1-\delta$ with
\begin{equation}
  \delta\le[1-\epsilon \nu(\Omega)]^N\,,
  \label{eq:failureProbability}
\end{equation}
where $\nu(\Omega):=1-\lambda_2(\Omega)$ denotes the spectral gap
between the largest and the second largest eigenvalues of $\Omega$
\cite{Pallister.etal2018, Zhu.Hayashi2019c}.
Hence, the QSV protocol $\Omega$ can verify the target state to fidelity
$1-\epsilon$ and confidence level $1-\delta$ with the number of copies of the
quantum states satisfying
\begin{equation}\label{eq:QSVparameter}
  N\geq\frac{\ln\delta^{-1}}{\ln\bigl\{[1-\nu(\Omega)\epsilon]^{-1}\bigr\}}\approx
  \frac1{\nu(\Omega)}\epsilon^{-1}\ln\delta^{-1}\,.
\end{equation}

In this work, we propose the concept of \emph{quantum process verification} (QPV).
Specifically, with a spirit similar to that for QSV, two efficient and practical
protocols are established for QPV.
Thanks to the Choi-Jamio{\l}kowski isomorphism \cite{Jamiolkowski1972, Choi1975},
we are able to relate QPV to QSV, and derive an ancilla-assisted protocol for QPV.
Then, we show that the ancilla-assisted protocol can also be transformed to the
prepare-and-measure protocol, which requires no ancilla systems.
Specifically, we demonstrate the efficacy of our protocols with the verification of
various quantum gates as well as the processes of some well-known quantum
circuits. In this way, we demonstrate that our protocols can provide an exponential
improvement over QPT and a quadratic improvement over DFE.
Moreover, these protocols are readily applicable with current experimental techniques
as only local measurements are required.
Last but not least, we show that our protocols for verifying quantum processes
can be easily adapted to the verification of quantum measurements.

\section{Choi-Jamio{\l}kowski isomorphism}%
Consider the (unnormalized) maximally entangled bipartite state
$\ket{\psi}=\sum_{k=1}^{d}\ket{k}_A\ket{k}_S$
between a quantum system $S$ and an ancilla system $A$,
where $\{\ket{k}\}_{k=1}^d$ represents an orthonormal basis.
For a quantum process $\cE$ acting only on the system $S$ of $\ket{\psi}$,
the output state is given by
\begin{equation}\label{eq:ChoiMat}
  \Upsilon_{\cE}=\bigl(\cI\otimes\cE\bigr)\bigl(\ket{\psi}\bra{\psi}\bigr)=
  \sum_{k,l=1}^{d}\ket{k}\bra{l}\otimes\cE\bigl(\ket{k}\bra{l}\bigr)\,,
\end{equation}
which is also called the Choi matrix of the process $\cE$.
Given the Choi matrix $\Upsilon_{\cE}$, the process $\cE$ can be obtained as
\begin{equation}\label{eq:ChoiRep}
  \cE(\rho)=\tr_{A}\bigl[(\rho^T\otimes\openone)\Upsilon_{\cE}\bigr]\,.
\end{equation}
The relations in Eqs.~\eqref{eq:ChoiMat} and \eqref{eq:ChoiRep} are known as
the Choi-Jamio{\l}kowski isomorphism \cite{Jamiolkowski1972, Choi1975},
an isomorphism between the Choi matrix ${\Upsilon_{\cE}\in\cB(\cH)\otimes
\cB(\cH)}$, and the linear map ${\cE\!\!: {\cB(\cH)\to\cB(\cH)}}$.
The Choi-Jamio{\l}kowski isomorphism also
implies that $\cE$ is a completely positive map if and only if the Choi matrix
$\Upsilon_{\cE}$ is positive semidefinite.
Note that here we do not require $\cE$ to be trace preserving. The benefit of
relaxing this restriction is that we can deal with the situation when
postselection or particle losses are allowed; see Sec.~\ref{sec:nonQPV} for more details.

Tomographically, Eq.~\eqref{eq:ChoiRep} implies that once the Choi matrix
$\Upsilon_{\cE}$ is determined, all the information on the process $\cE$ is
also obtained. This inspired the so-called ancilla-assisted QPT
\cite{DAriano.LoPresti2001,Altepeter.etal2003}, which requires additional 
ancilla systems but only fixed entangled input states.

\section{Ancilla-assisted quantum process verification}%
Similarly to the case of QPT, by making use of the Choi-Jamio{\l}kowski
isomorphism, one can verify the quantum process $\cK$ indirectly by verifying
the corresponding Choi state
${\rho_{\cK}:={\Upsilon_{\cK}}/{\tr(\Upsilon_{\cK})}}$ instead.  Especially,
when $\rho_{\cK}$ is pure, we can apply the QSV protocols for verifying $\cK$.
For simplicity, we first consider the case where $\cK$ is a unitary gate
$U$, i.e., the verification of quantum gates or quantum circuits.
Furthermore, we employ the entanglement gate fidelity
\begin{equation}
  F_e(\cE, U):=F(\rho_{\cE},\rho_{U})=\tr(\rho_{\cE}\rho_{U})
  \label{eq:fidelity}
\end{equation}
as a benchmark,
which is also directly related to the more widely-used notion, the average gate
fidelity by the relation \cite{Horodecki.etal1999},
\begin{equation}
  \bar{F}(\cE, U)=\frac{dF_e(\cE, U)+1}{d+1}\,.
  \label{eq:aveFidelity}
\end{equation}

Now, we are ready to formally define the QPV problem for quantum gates. Suppose
we have a quantum process $\cE$ which is promised to be a gate $U$. We want to
use the hypothesis testing method to verify this claim with a high confidence,
$1-\delta$.  In the ideal case, we want to distinguish two cases, $F_e(\cE,
U)=1$ and $F_e(\cE, U)\le 1-\epsilon$. The figure of merit is the sample
complexity, i.e., how many copies of input states are required with respect to
the infidelity $\epsilon$ and the confidence level $1-\delta$.

Due to the Choi-Jamio{\l}kowski isomorphism, we can transform the verification
of $U$ to the verification of the pure Choi state $\rho_U$. Furthermore, it can be
easily shown that $\rho_U$ is maximally entangled. Hence, with either
nonadaptive or adaptive QSV strategies $\Omega$
\cite{Hayashi.etal2006,Zhu.Hayashi2019a,Yu.etal2019,Li.etal2019}, we can achieve
the spectral gap
given by $\nu(\Omega)=(d-1)/d$. In practice, we may have more restrictions to
the allowed measurements, e.g., each party should be measured locally for verifying the multi-qubit gates.
However, according to the result in QSV
\cite{Pallister.etal2018}, we can show that the worst-case failure probability
in each run is always bounded by
\begin{equation}
  \max_{F(\rho_\cE,\rho_U)\le 1-\epsilon}\tr(\Omega\rho_\cE)
  \le 1-\epsilon\nu(\Omega)\,.
  \label{eq:errorAA}
\end{equation}
This implies the following proposition for the efficiency of this
ancilla-assisted protocol.

\begin{proposition}\label{thm:AAPV}
  For any quantum gate $U$, we can verify it to the entanglement gate fidelity
  $1-\epsilon$ and confidence level $1-\delta$ with
  $N\approx\frac1{v(\Omega)}\epsilon^{-1}\ln\delta^{-1}$ input states
  by verifying the corresponding Choi state $\rho_U$ with the QSV protocol
  $\Omega$.
\end{proposition}

We note that the inequality (instead of equality) in Eq.~\eqref{eq:errorAA}
results from the important difference between the general QSV and the QSV of
Choi states that when $\cE$ is restricted to trace-preserving processes,
$\rho_\cE$ also has the extra restriction that $\tr_S(\rho_\cE)=\openone/d$. This
restriction makes it possible to further improve the efficiency. However, if we
relax the restriction of trace-preserving, the inequality is always
attained; more details are reported in Sec.~\ref{sec:nonQPV}.  In fact, the above discussion
provides a general method to verify all the quantum processes whose
corresponding Choi states are pure.  We call this approach, which verifies the
quantum processes indirectly by QSV of the corresponding Choi states,
\textit{ancilla-assisted quantum process verification} (AAPV).

\section{Prepare-and-measure quantum process verification}%
The AAPV approach is easy to understand and straightforward to use with
the help of QSV.  However, the double requirements of additional ancilla
systems and maximally entangled input states are sometimes difficult to
achieve in experiments. This difficulty can be avoided by considering
the prepare-and-measure protocol. More precisely, we show that one
can always convert a one-way adaptive QSV protocol
(with nonadaptive QSV as a special case)
to aprepare-and-measure QPV (PMPV) protocol without any efficiency loss.

In a PMPV protocol, we randomly choose an input state $\rho_i$ with probability
$p_i$ and test the output state with the measurement $\{N_i, \openone-N_i\}$.
If the measurement outcome is $N_i$, then we say that channel $\cE$ passes the
test; otherwise we say that $\cE$ fails the test. As in QSV, we require that
the target gate $U$ always passes the test, i.e.,
\begin{equation}
  \tr(U\rho_iU^\dagger N_i)=1\,.
  \label{eq:UPass}
\end{equation}
For convenience, we denote the PMPV protocol 
\begin{equation}
  \Xi=\sum_{i}p_i\rho_i^T\otimes N_i\,.
  \label{eq:PMPV}
\end{equation}
Then in each run the worst-case failure probability is given by
\begin{equation}
  \max_{F_e(\cE,U)\le 1-\epsilon}\sum_ip_i\tr[\cE(\rho_i)N_i]
  =\max_{F_e(\cE,U)\le 1-\epsilon}\tr(\Xi\Upsilon_\cE)\,.
  \label{eq:efficiencyQPV}
\end{equation}
Given that $N_i\le\openone$, Eq.~\eqref{eq:UPass} is equivalent to
\begin{equation}
  \tr(\Xi\Upsilon_U)=1\,.
  \label{eq:UPassXi}
\end{equation}

For the normalized Choi state $\rho_U=\frac{1}{d}\Upsilon_{U}$,
the one-way adaptive QSV protocol
(classical communication from $A$ to $S$, with no communication, i.e., nonadaptive
QSV as a special case) takes on the general form
\cite{Yu.etal2019}
\begin{eqnarray}\label{eq:oneWayQSV}
  \Omega = \sum_{i} M_i \otimes N_i\,,
\end{eqnarray}
such that $\{M_i\}_i$ is a positive operator-valued measure (POVM) on the
ancilla system $A$, i.e., $\sum_iM_i=\openone$ and $\{N_i, \openone-N_i\}$ is
a pass-or-fail test on system $S$ which depends on the measurement outcome of
$\{M_i\}_i$. To ensure that the target state always passes the test, $\Omega$
must satisfy
\begin{equation}
  \tr(\Omega\rho_U)=\frac{1}{d}\tr(\Omega\Upsilon_{U})=1\,.
  \label{eq:oneWayPass}
\end{equation}

Now, we can convert any one-way adaptive QSV protocol for $\rho_U$ in
Eq.~\eqref{eq:oneWayQSV} to a PMPV protocol of the form of Eq.~\eqref{eq:PMPV} by letting
\begin{equation}
  p_i=\frac{\tr(M_i)}{d},~~\rho_i=\frac{M_i^T}{\tr(M_i)}\,,
  \label{eq:oneWayInput}
\end{equation}
where $\sum_ip_i=1$ follows from $\sum_iM_i=\openone$. Further, one can
easily verify that
\begin{equation}
  \Xi=\frac{1}{d}\Omega,~~\tr(\Xi\Upsilon_U)=1
  \label{eq:Xi}
\end{equation}
by considering Eqs.~\eqref{eq:oneWayPass} and \eqref{eq:oneWayInput}.
Then the following proposition characterizes the resources required in
the derived PMPV protocol.
\begin{proposition}\label{thm:PMPV}
  For any quantum gate $U$, the one-way ($A\to S$) adaptive AAPV protocol
  $\Omega$ can always be converted to a PMPV protocol $\Xi$ with the
  worst-case failure probability satisfying
  \begin{equation}
    \max_{F_e(\cE,U)\le 1-\epsilon}\tr(\Xi\Upsilon_\cE)
    =\max_{F(\rho_\cE,\rho_U)\le 1-\epsilon}\tr(\Omega\rho_\cE)
    \le 1-\epsilon\nu(\Omega)\,.
    \label{eq:error}
  \end{equation}
  Hence, for verifying the quantum gate $U$, the efficiency of the deduced PMPV
  protocol $\Xi$ is equal to the efficiency of the AAPV protocol $\Omega$.
\end{proposition}

A few remarks on the difference between $\Xi$ and $\Omega$ are in order.
First, although the AAPV $\Omega$ requires an ancilla system, it has the benefit that
only one single kind of (despite entangled) input state is needed. On the contrary,
in PMPV $\Xi$, no ancilla system is required, but many different kinds of input states are
needed.  Second, in the AAPV protocol $\Omega=\sum_{i} M_i \otimes N_i$, the POVM 
$\{M_i\}_i$ is usually preferred to be written as a convex combination of
projective measurements which are more experiment-friendly; this is not
necessary for PMPV, however, as $\rho_i={M_i^T}/{\tr(M_i)}$ are just different kinds of
input states.  For example, when no extra restrictions are imposed, instead of
using mutually unbiased bases as the input states, one can use the general
spherical $2$-designs, which are more easily constructed for a general
$d$-dimensional space \cite{Zhu.Hayashi2019a}.

\section{Applications}%
Quantum gates are basic yet essential components in various quantum information
processing tasks. Here we demonstrate that our schemes can verify various
quantum gates and quantum circuits efficiently and practically.
For demonstration, we take the verification of the {\sc cnot} gate, the Clifford
circuits, and the $n$-qubit $C^{(n-1)}Z$ and $C^{(n-1)}X$
gates as examples; more applications can be found in Appendix~\ref{app:MoreApp}.

The {\sc cnot} gate operates on two qubits, which flips the second qubit if and
only if the first qubit is $\ket{1}$.
To verify it, we use the four-qubit entangled state
\begin{eqnarray}
  \ket{\psi}=\ket{0000}+\ket{0101}+\ket{1010}+\ket{1111}
\end{eqnarray}
as input in the AAPV protocol.
Then the corresponding Choi matrix is given by
\begin{equation}
  \Upsilon_{\text{\sc CNOT}}=\bigl(\cI\otimes\cC_{\text{\sc
  NOT}}\bigl)\bigl(\ket{\psi}\bra{\psi}\bigl)=\ket{\phi}\bra{\phi}\,,
\end{equation}
where $\cC_{\text{\sc NOT}}$ denotes the corresponding operation of the {\sc
cnot} gate and
\begin{equation}
  \ket{\phi}=\ket{0000}+\ket{0101}+\ket{1110}+\ket{1011}\,.
\end{equation}
One notes that $\ket{\phi}$ is a stabilizer state, which can be verified efficiently by
\begin{equation}\label{eq:QSV_CNOT}
  \Omega_{\Upsilon_{\text{\sc CNOT}}}=\frac{1}{4}\bigl(P_{ZXZX}^{+}+P_{IZZZ}^{+}+P_{ZZIZ}^{+}+P_{XXXI}^{+}\bigl)\,,
\end{equation}
where $X$ and $Z$ are Pauli operators, and the superscript $+$ indicates the
projector onto the eigenspace with eigenvalue 1.
The spectral gap is given by $\nu(\Omega_{\Upsilon_{\text{\sc CNOT}}})=1/4$.
Accordingly, we can construct the PMPV protocol by employing the relation in
Eq.~\eqref{eq:Xi}.

The efficiency of verifying the \textsc{cnot} gate can be further improved by employing
more measurement settings. In fact, together with the Hadamard gate and the
phase gate, the {\sc cnot} gate generates the so-called
Clifford circuits which are key components in many schemes for quantum error correction
and become universal for quantum computation when augmented with certain state preparations.
Here, we show that our schemes are able to verify an arbitrary Clifford circuit
efficiently; see the proposition below.
\begin{proposition}
  For any $n$-qubit Clifford circuit, we can construct the AAPV and PMPV protocols
  for verifying it to the entanglement fidelity $1-\epsilon$ and the confidence
  level $1-\delta$ with the number of input states satisfying
  \begin{equation}
    N\approx\frac{2^{2n}-1}{2^{2n-1}}\epsilon^{-1}\ln\delta^{-1}\le2\epsilon^{-1}\ln\delta^{-1}\,.
  \end{equation}
  Furthermore, all the measurements required are local Pauli measurements.
\end{proposition}
\begin{proof}
  For the Clifford circuit $\cC$, the ancilla-assisted $\cI\otimes\cC$ remains
  as a Clifford circuit.
  Then the output Choi state is a $2n$-qubit stabilizer state since the input
  maximally entangled $\ket{\psi}=(\ket{00}+\ket{11})^{\otimes n}$ is
  a $2n$-qubit stabilizer state \cite{QCQI2010}.
  The verification of all stabilizer states \cite{Pallister.etal2018}
  can be constructed systematically by using the full set of stabilizers with
  $\nu(\Omega)={2^{2n-1}}/(2^{2n}-1)$, thus
  $N\approx\frac{2^{2n}-1}{2^{2n-1}}\epsilon^{-1}\ln\delta^{-1}
  \le2\epsilon^{-1}\ln\delta^{-1}$.
\end{proof}

We note that the best-known method so far for the verification of Clifford circuits is
DFE, which, however, requires $O(\epsilon^{-2}\log\delta^{-1})$ input states with
infidelity $\epsilon$ and confidence level $1-\delta$, so that our QPV protocols are
quadratically faster. In fact, from Propositions~\ref{thm:AAPV} and \ref{thm:PMPV},
the quadratic improvement of our QPV protocols over DFE with respect to the infidelity
$\epsilon$ is universal. In addition, although the physical settings of RB are
rather different, our QPV protocols are still quadratically more efficient
according to the statistical analysis in Refs.~\cite{Magesan.etal2011} and \cite{Wallman.Flammia2014}
if sample complexity is considered.

\begin{figure}
  \centering
  \includegraphics[width=.3\textwidth]{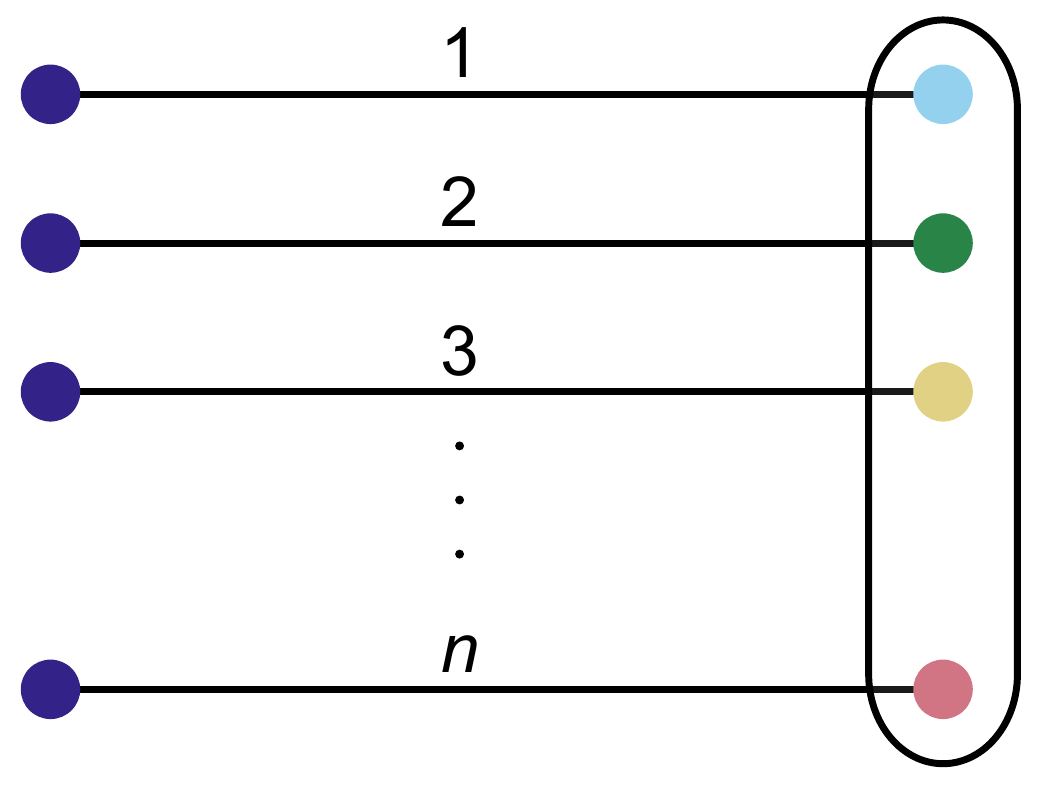}
  \caption{Hypergraph corresponding to the output Choi state
    $\openone^{\otimes n}\otimes C^{(n-1)}Z\ket{\tilde{\psi}}$.  An
  $(n+1)$-coloring of the hypergraph is also shown.}
  \label{fig:CZ}
\end{figure}

Besides the Clifford gates, our protocols can also be used for the efficient
verification of other important classes of quantum gates. For example, by taking
advantage of the verification method for hypergraph states
\cite{Zhu.Hayashi2019b}, we can verify $n$-qubit $C^{(n-1)}Z$ and $C^{(n-1)}X$
gates including the Toffoli (\textsc{ccnot}) gate, with
$N\approx(n+1)\epsilon^{-1}\ln\delta^{-1}$ input states.
Here, we only show the
verification of $C^{(n-1)}Z$ gates. The verification of $C^{(n-1)}X$ gates is
equivalent to that of $C^{(n-1)}Z$ gates up to some local unitary
transformations.

For the construction of the Choi state, instead of choosing the input state as
$\ket{\psi}=(\ket{00}+\ket{11})^{\otimes n}$, we use
\begin{equation}
  \ket{\tilde\psi}=(\ket{0}\ket{+}+\ket{1}\ket{-})^{\otimes n}
  =CZ^{\otimes n}\ket{+}^{\otimes 2n},
  \label{eq:inputCZ}
\end{equation}
which is equivalent to $\ket{\psi}$ up to a local unitary transformation.
Note also that we have omitted the normalization in the derivation of
Eq.~\eqref{eq:inputCZ}.  Then after the unitary gate $\openone^{\otimes
n}\otimes C^{(n-1)}Z$, the output Choi state is a hypergraph state as shown in
Fig.~\ref{fig:CZ}. One can easily see that the hypergraph is $n+1$ colorable,
as illustrated in the figure.
Thus, by using the coloring protocol proposed in Ref.~\cite{Zhu.Hayashi2019b},
we can construct the AAPV protocol with the spectral gap given by
$\nu(\Omega)=1/(n+1)$.
Hence, the number of input states required for verifying the $C^{(n-1)}Z$ gate
to the fidelity $1-\epsilon$ and the confidence level $1-\delta$ is
$N\approx(n+1)\epsilon^{-1}\ln\delta^{-1}$.
Similarly, one gets the PMPV protocol by using the relation in
Eq.~\eqref{eq:Xi}.

Last but not least, our method is also applicable to the verification of
the processes of some well-known quantum algorithms.
In Appendix~\ref{app:VofDJ}, we take the Deutsch-Jozsa algorithm
\cite{DeutschJozsa1992} as an example to illustrate this point.

\section{Verification of non-trace-preserving processes}\label{sec:nonQPV}
In real experiments, the quantum process $\cE$ is often not
trace preserving due to postselection or particle losses. In this section,
we show that our QPV protocols still work in such a scenario.

We consider the AAPV protocol first and define the Choi matrix as
\begin{equation}\label{eq:ChoiMatA}
  \Upsilon_{\cE}=\bigl(\cI\otimes\cE\bigr)\bigl(\ket{\psi}\bra{\psi}\bigr)=
  \sum_{k,l=1}^{d}\ket{k}\bra{l}\otimes\cE\bigl(\ket{k}\bra{l}\bigr)\,,
\end{equation}
where $\ket{\psi}=\sum_{k=1}^{d}\ket{k}_A\ket{k}_S$, and the
Choi-Jamio{\l}kowski isomorphism
\begin{equation}\label{eq:ChoiRepA}
  \cE(\rho)=\tr_{A}\bigl[(\rho^T\otimes\openone)\Upsilon_{\cE}\bigr]
\end{equation}
still holds. Further, we define the corresponding Choi state as
\begin{equation}
  \rho_{\cE}:={\Upsilon_{\cE}}/{\tr(\Upsilon_{\cE})}\,.
  \label{eq:ChoiState}
\end{equation}
When $\cE$ is trace preserving, $\rho_\cE$ can be interpreted as the output of
the process $\cI\otimes\cE$ for the (normalized) input state $\ket{\psi}$. When
$\cE$ is not trace preserving (more precisely, is trace decreasing), $\rho_\cE$
can be viewed as the output state after postselection. That is,
an output is not guaranteed to be obtained for every input.
Hence, we also call $\rho_\cE$
the postselected Choi state. We can apply the QSV protocol for verifying $\cK$
when $\rho_\cK$ is a pure state. For a quantum process $\cK$, $\rho_\cK$ is
a pure state if and only if
\begin{equation}
  \cK(\rho)=K\rho K^\dagger\,.
  \label{eq:postProcess}
\end{equation}
Without loss of generality, we assume that $K$ is of full rank; then the
corresponding operation with post-selection can be written as
\begin{equation}
  \tilde\cK(\rho)=\frac{K\rho K^\dagger}{\tr(K\rho K^\dagger)}\,.
  \label{eq:postProcessN}
\end{equation}
Note that under post-selection, $K$ can be verified only up to a constant
factor $c$, i.e., $\tilde{\cK}$ is invariant when ${K\to cK}$. Hence, instead of
choosing the number of copies of the input states as the figure of merit,
we choose the number of copies of the post-selected
output states, which is invariant under $K\to cK$.
Still, we rely on the entanglement fidelity defined as
\begin{equation}
  F_e(\cE, \cK):=F(\rho_{\cE},\rho_{\cK})\,.
  \label{eq:fidelityA}
\end{equation}
Furthermore, one can easily see that when we relax the trace-preserving
condition, $\rho_\cE$ can be any quantum state according to the
Choi-Jamio{\l}kowski isomorphism. Then, the following proposition follows
directly from the corresponding result in QSV.

\begin{proposition}\label{thm:AAPVpost}
  For any quantum process with postselection $\tilde\cK(\rho)=K\rho
  K^\dagger/\tr(K\rho K^\dagger)$, we can verify $\tilde{\cK}$
  to the entanglement fidelity $1-\epsilon$ and confidence level $1-\delta$ with
  $N\approx\frac1{v(\Omega)}\epsilon^{-1}\ln\delta^{-1}$ postselected output
  states by verifying the postselected Choi state $\rho_\cK$ with the QSV protocol
  $\Omega$.
\end{proposition}

As in the case of verification for quantum gates, we can show that a one-way (${A\to S}$)
adaptive AAPV protocol can be converted to a PMPV protocol without any
efficiency loss. Still, we denote a PMPV protocol 
\begin{equation}
  \Xi=\sum_{i}p_i\rho_i^T\otimes N_i\,,
  \label{eq:PMPVA}
\end{equation}
with the requirement that $\tr[\tilde\cK(\rho_i)N_i]=1$, i.e.,
\begin{equation}
  \tr(K_i\rho_i K_i^\dagger N_i)=\tr(K_i\rho_iK_i^\dagger)\,.
  \label{eq:UPassA}
\end{equation}
Then in each run (when there is an output state), the worst-case failure
probability is given by
\begin{equation}
  \max_{F_e(\cE,\cK)\le 1-\epsilon}
  \frac{\sum_ip_i\tr[\cE(\rho_i)N_i]}{\sum_ip_i\tr[\cE(\rho_i)]}
  =\max_{F_e(\cE,\cK)\le 1-\epsilon}
  \frac{\tr(\Xi\Upsilon_\cE)}{\tr[\cE(\bar{\rho})]}\,,
  \label{eq:efficiencyQPVA}
\end{equation}
where $\bar{\rho}:=\sum_ip_i\rho_i$.  Given that $N_i\le\openone$,
Eq.~\eqref{eq:UPassA} is equivalent to 
\begin{equation}
  \frac{\tr(\Xi\Upsilon_\cK)}{\tr[\cK(\bar{\rho})]}=1\,.
  \label{eq:UPassXiA}
\end{equation}

For the Choi state $\rho_\cK$, the one-way (${A\to S}$) adaptive QSV protocol
takes on the general form
\begin{eqnarray}\label{eq:oneWayQSVA}
  \Omega = \sum_{i} M_i \otimes N_i\,,
\end{eqnarray}
such that $\{M_i\}_i$ is a POVM on the ancilla system $A$, i.e.,
$\sum_{i=1}^nM_i=\openone$ and $\{N_i, \openone-N_i\}$ is a pass-or-fail test of
system $S$ which depends on the measurement outcome of $\{M_i\}_i$. To ensure
that the target state always passes the test, $\Omega$ also must satisfy
\begin{equation}
  \tr(\Omega\rho_\cK)=\frac{\tr(\Omega\Upsilon_\cK)}{\tr(\Upsilon_\cK)}=1\,.
  \label{eq:oneWayPassA}
\end{equation}

Now, we can convert any one-way adaptive QSV for $\rho_\cK$ in
Eq.~\eqref{eq:oneWayQSVA} to a PMPV of the form of Eq.~\eqref{eq:PMPVA} by letting
\begin{equation}
  p_i=\frac{\tr(M_i)}{d},~~\rho_i=\frac{M_i^T}{\tr(M_i)}\,,
  \label{eq:oneWayInputA}
\end{equation}
where $\sum_ip_i=1$ follows from $\sum_iM_i=\openone$. Further, we have
\begin{equation}
  \bar\rho=\sum_ip_i\rho_i=\frac{\openone}{d},~~
  \tr(\Upsilon_\cK)=\tr[\cK(\openone)]=d\tr[\cK(\bar\rho)]\,.
  \label{eq:rhobar}
\end{equation}
Then, one can easily verify that
\begin{equation}
  \Xi=\frac{1}{d}\Omega,~~
  \frac{\tr(\Xi\Upsilon_\cK)}{\tr[\cK(\bar{\rho})]}=1\,.
\end{equation}
from Eqs.~\eqref{eq:oneWayPassA}, \eqref{eq:oneWayInputA}, and
\eqref{eq:rhobar}.
Similarly, we can also show that
\begin{equation}
  \tr[\cE(\bar\rho)]=\tr[\cE(\frac{\openone}{d})]=\frac{1}{d}\tr(\Upsilon_\cE)\,.
\end{equation}
Thus, we have the following proposition for the resources required in the
derived PMPV protocol.
\begin{proposition}\label{thm:PMPVA}
  For any quantum process with postselection $\tilde\cK(\rho)=K\rho
  K^\dagger/\tr(K\rho K^\dagger)$, the one-way ${(A\to S)}$ adaptive QSV protocol
  $\Omega$ can always be converted to a PMPV protocol $\Xi$ with the worst-case
  failure probability satisfying
  \begin{equation}
    \max_{F_e(\cE,\cK)\le 1-\epsilon}
    \frac{\tr(\Xi\Upsilon_\cE)}{\tr[\cE(\bar{\rho})]}
    =\max_{F(\rho_\cE,\rho_\cK)\le 1-\epsilon}\tr(\Omega\rho_\cE)
    =1-\epsilon\nu(\Omega)\,.
    \label{eq:errorA}
  \end{equation}
  Hence, to verify the quantum process with postselection
  $\tilde\cK(\rho)=K\rho K^\dagger/\tr(K\rho K^\dagger)$, the efficiency of the
  deduced PMPV protocol $\Xi$ is equal to the efficiency of the AAPV protocol
  $\Omega$.
\end{proposition}

\section{Verification of quantum measurements}%
Finally, we briefly show that our protocols for process verification can
also be easily adapted to verify quantum measurements. The verification of
quantum measurements is similar to the verification of quantum processes with
postselection. However, they are also different in the sense that we can no longer
assume the availability of reliable quantum measurements for
measurement verification.

Suppose that we want to verify the projective measurement
$\cP=\{\ket{i}\bra{i}\}_{i=1}^d$. For an arbitrary measurement
$\cM=\{M_i\}_{i=1}^d$, we characterize the fidelity between $\cP$ and $\cM$
with
\begin{equation}
  F(\cM,\cP)=\frac{1}{d}\sum_{i=1}^d\bra{i}M_i\ket{i}\,.
  \label{eq:fidelityMeas}
\end{equation}
If we want to distinguish the two cases, $F(\cM, \cP)=1$ and $F(\cM, \cP)\le
1-\epsilon$, we just need to prepare the input state $\ket{i}$ with probability $1/d$
and measure it with the measurement $\cM$. If the measurement outcome is $i$,
then we say that $\cM$ passes the test; otherwise we say that $\cM$ fails the
test. Thus, the failure probability of the protocol in each run is given by
\begin{equation}
  \max_{F(\cM,\cP)\le 1-\epsilon}
  \frac{1}{d}\sum_{i=1}^d\bra{i}M_i\ket{i}=1-\epsilon\,,
  \label{eq:failMeas}
\end{equation}
which follows directly from the definition in Eq.~\eqref{eq:fidelityMeas}.
This implies that we can verify the projective measurement $\cP$ with
$N\approx\epsilon^{-1}\ln
\delta^{-1}$ copies of input states.

For entangled measurements in multipartite systems, we have two choices.
First, if we can generate entangled input states, we can simply use the
previous protocol for verifying the measurements. Second, in the case where we can trust
the reliability of single-partite measurements, then the verification of entangled
measurements (or, more precisely, quantum instruments) can
be treated as verifying quantum processes with post-selection
(see Sec.~\ref{sec:nonQPV}).

\section{Summary}%
The efficient characterization and identification of quantum processes play
a crucial role in almost all tasks of quantum information processing, such as
quantum computation, quantum communication, quantum metrology, and more.
In this work, by proposing the concept of QPV, we have established two
efficient and practical protocols for verifying quantum processes based on a spirit similar to that for QSV. Compared to the known methods, our protocols
can provide an exponential improvement over QPT and a quadratic improvement over DFE.
Specifically, we demonstrate the efficacy of our verification protocols
with many applications, including the verification of various quantum gates,
quantum circuits, and processes of quantum algorithms. As an
extension, we prove that quantum measurements can also be efficiently verified
by using an idea similar to that for quantum processes.  Moreover, the protocols
proposed in this work are well within the reach of current experimental
techniques, as only local measurements are needed.  As an outlook, it is
meaningful to consider how our protocols should be modified in the presence of
state preparation and measurement errors.

\acknowledgments
We are grateful to Otfried G\"uhne, Zhen-Peng Xu, and Hui Khoon Ng for helpful discussions.
J.S. especially thanks Liuyue Shang for being quiet sometimes.
This work was supported by the National Key R\&D Program of China
under Grant No.~2017YFA0303800 and the National Natural Science Foundation of
China through Grant Nos.~11574031, 61421001, and 11805010.  J.S. also
acknowledges support by the Beijing Institute of Technology Research Fund
Program for Young Scholars. X.D.Y. acknowledges support by the DFG and the ERC
(Consolidator Grant No.~683107/TempoQ).

\textit{Note added.---}%
Recently, we became aware of related works
by Zhu et al. \cite{Zhu.Zhang2019} and Zeng et al. \cite{Zeng.etal2019} .

\appendix

\section{More applications}\label{app:MoreApp}
In this Appendix, we present the verification of more quantum gates using
the AAPV and PMPV protocols.

\subsection{Identity}
The identity $\cI(\rho)=\rho$ represents a trivial process.
The corresponding Choi matrix is the maximally entangled state itself, i.e.,
\begin{equation}\label{eq:ChoiState_I}
\Upsilon_{\cI}=\ket{\psi}\bra{\psi}=\left[ \begin{smallmatrix}
1& 0 & 0& 1 \\0& 0 & 0& 0 \\0& 0 & 0& 0 \\ 1& 0 & 0& 1  \end{smallmatrix} \right].
\end{equation}
Then, the AAPV protocol to verify $\cI$ is given by
\begin{eqnarray}\label{eq:QSV_I}
  \Omega_{\Upsilon_{\cI}}=\frac{1}{2}\bigl(P_{XX}^{+}+P_{ZZ}^{+}\bigl)\,,
\end{eqnarray}
and the spectral gap is $\nu(\Omega_{\Upsilon_{\cI}})=1/2$.
Accordingly, the PMPV protocol is given by
\begin{align}\label{eq:QPV_I}
  \Xi_{\Upsilon_{\cI}}=&\frac{1}{4}\bigl(\,\ket{+}\bra{+}\otimes\ket{+}\bra{+}
  +\ket{-}\bra{-}\otimes\ket{-}\bra{-}\nonumber\\
  &+\ket{0}\bra{0}\otimes\ket{0}\bra{0}
  +\ket{1}\bra{1}\otimes\ket{1}\bra{1}\,\bigl)\,,
\end{align}
where $\ket{\pm}=\frac{1}{\sqrt{2}}(\ket{0}\pm\ket{1})$.
The verification operator $\Xi_{\Upsilon_{\cI}}$ tells us all the prepare-and-measure operators with
their corresponding probability distributions.

Note that the protocol $\Omega_{\Upsilon_{\cI}}$ in Eq.~\eqref{eq:QSV_I} is not actually optimal.
The efficiency can be further improved by using more measurement settings such that
\begin{eqnarray}\label{eq:QSV_Iopt}
  \Omega'_{\Upsilon_{\cI}}=\frac{1}{2}\bigl(P_{XX}^{+}+P_{YY}^{-}+P_{ZZ}^{+}\bigl)\,,
\end{eqnarray}
and then the spectral gap is $\nu(\Omega'_{\Upsilon_{\cI}})=2/3$.
Similar results can be obtained for the PMPV protocol, where the efficiency
can be improved by involving more different kinds of input states.

Furthermore, the above protocols can be directly adapted for verifying any
single-qubit gate $U$ by the following substitutions:
\begin{equation}
  \begin{aligned}
  X\otimes X&\to X\otimes UXU^\dagger,\\
  Y\otimes Y&\to Y\otimes UYU^\dagger,\\
  Z\otimes Z&\to Z\otimes UZU^\dagger.
  \end{aligned}
\end{equation}
See below for some concrete examples.

\subsection{The bit-flip gate}%
The bit-flip gate is $X=\left[ \begin{smallmatrix} 0& 1 \\ 1 & 0 \end{smallmatrix} \right]$,
which is actually the Pauli-$X$ operator.
The corresponding Choi matrix is
\begin{equation}\label{eq:ChoiState_X}
  \Upsilon_{X}
  =\bigl(\cI \otimes \cX\bigl)\bigl(\ket{\psi}\bra{\psi}\bigl)
  =\left[ \begin{smallmatrix} 0& 0 & 0& 0 \\0& 1 & 1& 0 \\0& 1 & 1& 0 \\ 0& 0 & 0& 0 \end{smallmatrix} \right].
\end{equation}
Then, the AAPV protocol to verify $X$ is given by
\begin{eqnarray}\label{eq:QSV_X}
  \Omega_{\Upsilon_{X}}&=&\frac{1}{2}\bigl(P_{XX}^{+}+P_{ZZ}^{-}\bigl)\,,
\end{eqnarray}
with the spectral gap being $\nu(\Omega_{\Upsilon_{X}})=1/2$.
Accordingly, the PMPV protocol is given by
\begin{align}\label{eq:QPV_X}
  \Xi_{\Upsilon_{X}}=&\frac{1}{4}\bigl(\,\ket{+}\bra{+}\otimes\ket{+}\bra{+}
  +\ket{-}\bra{-}\otimes\ket{-}\bra{-}\nonumber\\
  &+\ket{1}\bra{1}\otimes\ket{0}\bra{0}
  +\ket{0}\bra{0}\otimes\ket{1}\bra{1}\,\bigl)\,.
\end{align}

\subsection{The Hadamard gate}%
The Hadamard gate is $H=\frac1{\sqrt{2}}\left[ \begin{smallmatrix} 1& 1 \\ 1 & -1 \end{smallmatrix} \right]$,
with the corresponding Choi matrix
\begin{equation}\label{eq:ChoiState_H}
  \Upsilon_{H}
  =\bigl(\cI \otimes \cH \bigl)\bigl(\ket{\psi}\bra{\psi}\bigl)
  =\frac1{2}\left[ \begin{smallmatrix} 1& 1 & 1& -1 \\1& 1 & 1& -1 \\1& 1 & 1& -1 \\ -1& -1 & -1& 1  \end{smallmatrix} \right].
\end{equation}
Then, the AAPV protocol is given by
\begin{eqnarray}\label{eq:AAPV_H}
  \Omega_{\Upsilon_{H}}&=&\frac{1}{2}\bigl(P_{XZ}^{+}+P_{ZX}^{+}\bigl)\,,
\end{eqnarray}
with the spectral gap being $\nu(\Omega_{\Upsilon_{H}})=1/2$.
Accordingly, the PMPV protocol is given by
\begin{align}\label{eq:PMPV_H}
  \Xi=&\frac{1}{4}\bigl(\,\ket{+}\bra{+}\otimes\ket{0}\bra{0}
  +\ket{-}\bra{-}\otimes\ket{1}\bra{1}\nonumber\\
  &+\ket{0}\bra{0}\otimes\ket{+}\bra{+}
  +\ket{1}\bra{1}\otimes\ket{-}\bra{-}\,\bigl)\,.
\end{align}

\subsection{The phase gate}%
The one-qubit phase gate is written as $S=\left[ \begin{smallmatrix} 1& 0 \\ 0 & \I \end{smallmatrix} \right]$,
with the corresponding Choi matrix
\begin{equation}\label{eq:ChoiState_S}
  \Upsilon_{S}
  =\bigl(\cI \otimes \cS \bigl)\bigl(\ket{\psi}\bra{\psi}\bigl)
  =\left[ \begin{smallmatrix} 1& 0 & 0& -\I \\0& 0 & 0& 0 \\0& 0 & 0& 0 \\ \I& 0 & 0& 1  \end{smallmatrix} \right].
\end{equation}
Then, the AAPV protocol is given by
\begin{eqnarray}\label{eq:AAPV_S}
  \Omega_{\Upsilon_{S}}&=&\frac{1}{2}\bigl(P_{ZZ}^{+}+P_{XY}^{+}\bigl)\,,
\end{eqnarray}
with the spectral gap being $\nu(\Omega_{\Upsilon_{S}})=1/2$.
Accordingly, the PMPV protocol is given by
\begin{align}\label{eq:PMPV_S}
  \Xi_{\Upsilon_{S}}=&\frac{1}{4}\bigl(\,\ket{0}\bra{0}\otimes\ket{0}\bra{0}
  +\ket{1}\bra{1}\otimes\ket{1}\bra{1}\nonumber\\
  &+\ket{+}\bra{+}\otimes\ket{\top}\bra{\top}
  +\ket{-}\bra{-}\otimes\ket{\bot}\bra{\bot}\,\bigl),
\end{align}
where $\ket{\top}=\frac1{\sqrt{2}}(\ket{0}+\I\ket{1})$ and $\ket{\bot}=\frac1{\sqrt{2}}(\I\ket{0}+\ket{1}$).

\section{Verification of the Deutsch-Jozsa algorithm}\label{app:VofDJ}
The Deutsch-Jozsa algorithm \cite{DeutschJozsa1992} is used to determine
whether a given function $f(x)$ is \textit{constant} or \textit{balanced}.
The process of this algorithm is different depending on the function type of $f(x)$.
Below we show the verification of a few typical cases using our protocols.

First, we consider the case where $f(x)$ is a constant function.
If $f(x)=0$, the process is the trivial identity $\cI^{\otimes(n+1)}$.
If instead $f(x)=1$, the process of the algorithm is $\cI^{\otimes
n}\otimes\cZ$, which can also be efficiently verified.
For demonstration, we consider the three-qubit case,
i.e., $\cI^{\otimes 2}\otimes\cZ$. The corresponding Choi matrix is given by
\begin{equation}
  \Upsilon_{f(x)=1}=\ket{\phi_1}\bra{\phi_1}\,,
\end{equation}
where
\begin{align}
  \ket{\phi_1}=&\ket{000000}-\ket{001001}+\ket{010010}-\ket{011011}\nonumber\\
	      +&\ket{100100}-\ket{101101}+\ket{110110}-\ket{111111}\,.
\end{align}
Then, we can construct the AAPV protocol as
\begin{align}
  \Omega_{\Upsilon_{f(x)=1}} 
  = &\frac1{6}\bigl(P_{ZZZZZZ}^{+}+P_{ZZYZZY}^{+}+P_{ZZXZZX}^{-}\nonumber\\
		    +&P_{ZYZZYZ}^{-}+P_{ZXZZXZ}^{+}+P_{YZZYZZ}^{-}\bigl)\,,
\end{align}
which uses six stabilizer generators.
The spectral gap is $\nu(\Omega_{\Upsilon_{f(x)=1}})=1/6$,
so the number of input states required is
$N\approx6\epsilon^{-1}\ln\delta^{-1}$ with fidelity $1-\epsilon$ and confidence level $1-\delta$.

Next, we consider the balanced case, which may have many different varieties of
$f(x)$. Take the three-qubit case $\cI\otimes\cC_{\text{\sc NOT}}$ as an example,
in which the function is completely determined by the second qubit, namely,
$f(x)=x_2$. The corresponding Choi matrix is
\begin{equation}
  \Upsilon_{f(x)=x_2}=\ket{\phi_b}\bra{\phi_b}\,,
\end{equation}
where
\begin{align}
  \ket{\phi_b}  
  = &\ket{000000}+\ket{001001}+\ket{011010}+\ket{010011}\nonumber\\
    +&\ket{100100}+\ket{101101}+\ket{111110}+\ket{110111}\,.
\end{align}
Then, the AAPV protocol is constructed as
\begin{align}
  &\Omega_{\Upsilon_{f(x)=x_2}} 
  =\frac1{6}\bigl(P_{ZZZZZZ}^{+}+P_{ZZYZZY}^{-}+P_{ZZXZZX}^{+}\nonumber\\
 &\qquad+P_{ZYZZYZ}^{-}+P_{ZXZZXZ}^{+}+P_{YZZYZZ}^{-}\bigl)\,,
\end{align}
which also utilizes six stabilizer generators, and the spectral gap is $\nu(\Omega_{\Upsilon_{f(x)=x_2}})=1/6$.
Accordingly, the PMPV protocols for all the cases studied above can be easily
constructed using Eq.~\eqref{eq:Xi}.


%

\end{document}